\documentclass{article}
\usepackage{mathtools}
\usepackage{verbatim}
\usepackage{float}
\usepackage{graphicx}
\usepackage{amssymb}
\usepackage{amsthm}
\usepackage{cite}
\usepackage{titling, lipsum}

\binoppenalty=\maxdimen
\relpenalty=\maxdimen

\newtheorem{proposition}{Theorem}[subsection]
\newtheorem{theorem}{Lemma}
\newtheorem{fact}{Fact}

\usepackage[parfill]{parskip}

\begin{document}

\clearpage
\begin{center}

\noindent

\rule{\linewidth}{0.5mm} \\[0.4cm]
{\huge \bfseries Continuous Time Gathering of Agents with Limited Visibility
and Bearing-Only Sensing \\[0.4cm]}

\rule{\linewidth}{0.5mm} \\[1.5cm]

\begin{center} \Large
Levi-Itzhak Bellaiche  and  Alfred Bruckstein                                         
\end{center}

\vfill

\begin{center} \large
Center for Intelligent Systems - MARS Lab \\
Autonomous Systems Program \\
Technion Israel Institute of Technology \\
32000 Haifa, Israel
\end{center}

\end{center}

\thispagestyle{empty}

\newpage

\setcounter{page}{1}

{\bf Abstract} \\

A group of mobile agents, identical, anonymous, and oblivious (memoryless),
having the capacity to sense only the relative direction (bearing) to
neighborhing agents within a finite visibility range, are shown to gather to a meeting point in finite time by applying a very simple rule of motion. The
agents' rule of motion is : set your velocity vector to be the sum of the two
unit vectors in $\mathbb{R}^2$ pointing to your ``extremal'' neighbours
determining the smallest visibility disc sector in which all your visible neighbors reside,
provided it spans an angle smaller than $\pi$, otherwise, since you are
``surrounded'' by visible neighbors, simply stay put (set your velocity to 0).
Of course, the initial constellation of agents must have a visibility graph that is
connected, and provided this we prove that the agents gather to a common meeting
point in finite time, while the distances between agents that initially see each
other monotonically decreases.

\newpage

\section{Introduction}

This paper studies the problem of mobile agent convergence, or robot gathering
under severe limitations on the capabilities of the agent-robots. We assume that the agents
move in the environment (the plane $\mathbb{R}^2$) according to what they
currently ``see'', or sense in their neighborhood. All agents are identical and
indistinguishable (i.e. they are anonimous having no i.d's)
and, all of them are performing the same ``reactive'' rule of motion in response to what they see. Our assumption
will be that the agents have a finite visibility range $V$, a distance beyond
which they cannot sense the presence of other agents. The agents
within the ``visibility disk'' of radius V around each agent are defined as his
neighbors, and we further assume that the agent can only sense
the direction to its neighbors, i.e. it performs a ``bearing only'' measurement
yielding unit vectors pointing toward its neighbor. \\
Therefore, in our setting, each agent senses its neighbors within the visibility
disk and sets its motion only according to the distribution of unit vectors
pointing to its current neighbors. Figure \ref{fig:constellation} shows a
constellation of agents in the plane ($\mathbb{R}^2$), their ``visibility graph'' and the visibility disks of
some of them, each agent moves based on the set of unit vectors
pointing to its neighbors.\\

\begin{figure}[H]
  \centering
  \includegraphics[width=0.7\textwidth,natwidth=787,natheight=780]{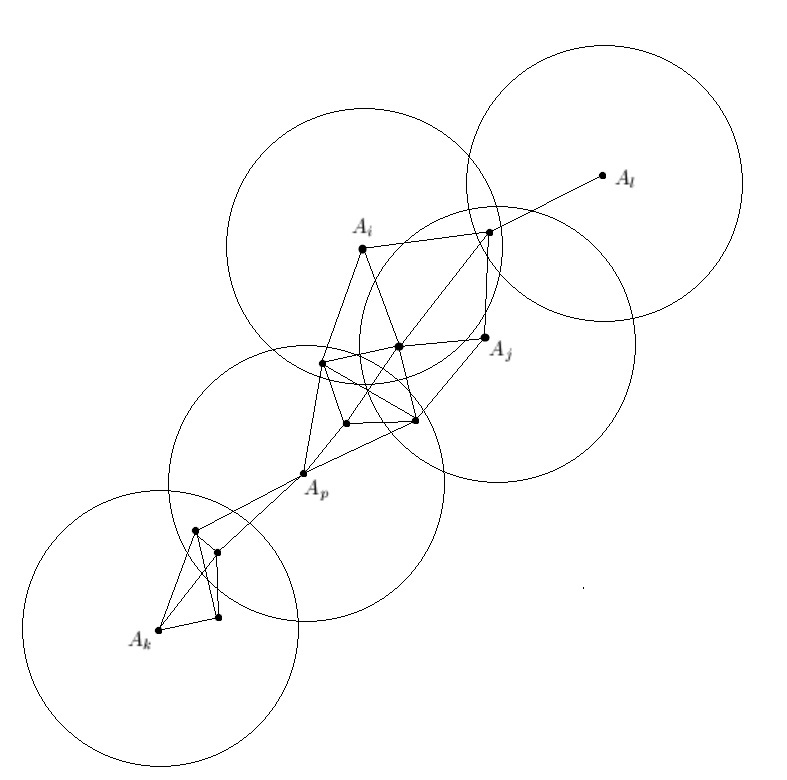}
  \caption{A constellation of agents in the plane displaying the ``visibility
  disks'' of agents $A_k, A_l,A_i,A_j,A_p$ and the visibility graph that they
  define, having edges connecting pairs of agents that can see each other.}
  \label{fig:constellation}
\end{figure}

In this paper we shall prove that continuous time limited visibility sensing of
directions only and continuous adjustment of agents' velocities according to
what they see is enough to ensure the gathering of the agents in finite time to
a point of encounter. \\
\newline
The literature of robotic gathering is vast and the problem was addressed under
various assumptions on the sensing and motion capabilities of the agents. Here
we shall only mention papers that deal with gathering assuming continuous time
motion and limited visibility sensing, since these are most relevant to our work
reported herein. The paper \cite{OlfatiSaber} by Olfati-Saber, Fox, and Murray,
nicely surveys the work on the topic of gathering (also called consensus) for networked multi agent systems, where the connections
between agents are not necessarily determined by their relative position or
distance. This approach to multi-agent systems was indeed the subject of much
investigation and some of the results, involving ``switching connection
topologies'' are useful in dealing with constellation-defined visibility-based
interaction dynamics too. A lot of work was invested in the analysis of
``potential functions'' based multi-agent dynamics, where agents are sensing
each other through a ``distance-based'' influence field, a prime example here
being the very influential work of Gazi and Passino \cite{VGazi} which analyses
beautifully the stability of a clustering process. Interactions involving hard limits on the ``visibility distance'' in
sensing neighbors were analysed in not too many works. Ji and
Eggerstedt in \cite{Eggerstedt} analysed such problems using potential functions
that are ``visibility-distance based barrier functions'' and proved connectedness-preservation properties at the expense of
making some agents temporarily ``identifiable'' and ``traceable'' via a
hysteresis process.
Ando, Oasa, Suzuki and Yamashita in \cite{Ando} were the first to deal with
hard constraints of limited visibility and analysed the ``point convergence'' or gathering issue in a discrete-time
synchronized setting, assuming agents can see and measure both distances
and bearing to neighbors withing the visibility range.\\
Subsequently, in a series of papers, Gordon, Wagner, and
Bruckstein, in \cite{Gordon1}, \cite{Gordon2}, \cite{Gordon3}, analysed
gathering with limited visibility and bearing only sensing constraints imposed on the agents. Their work proved gathering or clustering results in discrete-time
settings, and also proposed dynamics for the continuous-time settings. In the
sequel we shall mention the continuous time motion model they analysed and
compare it to our dynamic rule of motion. \\
In our work, as well as most of the papers mentioned above one assumed that the
agents can directly control their velocity with no acceleration constraints. We
note that the literature of multi-agent systems is replete with papers assuming
more complex and realistic dynamics for the agents, like unicycle motions,
second order systems and double integration models relating the location to the
controls, and seek sensor based local control-laws that ensure gathering or the
achievement of some desired configuration. However we feel that it is still worthwhile exploring systems with agents
directly controlling their velocity based on very primitive sensing, in order to
test the limits on what can be achieved by agents with such simple, reactive
behaviours. 

\section{The gathering problem}

We consider N agents located in the plane ($\mathbb{R}^2$) whose positions
are given by $\{ P_k = (x_k,y_k)^T\}_{k=1,2,\ldots,N}$, in some absolute
coordinate frame which is unknown to the agents.We define the vectors 

\begin{displaymath}
u_{ij} = \left\{ \begin{array}{ccc}
				\frac{P_j  -P_i}{\Vert P_j - P_i \Vert} & & 0 < \Vert P_j - P_i \Vert
				\leq V \\
				0 & & \Vert P_j - P_i \Vert = 0 \mbox{  or  } \Vert P_j - P_i \Vert > V
				\end{array}
			\right.
\end{displaymath}

hence $u_{ij}$ are, if not zero, the unit vectors from $P_i$ to all $P_j$'s
which are neighbors of $P_i$ in the sense of being at a distance less than $V$ from
$P_i$, i.e. $P_j$'s obeying :
\begin{displaymath}
\Vert P_j - P_i \Vert \triangleq [(P_j-P_i)^T (P_j-P_i)]^{1/2} \leq V
\end{displaymath}
Note that we have $u_{ij} = -u_{ji}, \forall (i,j)$. For each agent $P_i$, let
us define the special vectors, $u_i^{+}$ and $u_i^{-}$ (from among
the vectors $u_{ij}$ defined above). Consider the nonzero vectors from the set
$\{u_{ij}\}_{j=1,2,\ldots,N}$. Anchor a moving unit vector $\bar{\eta}(\theta)$
at $P_i$ poiting at some arbitrary neighbor, i.e. at $u_{ik} \neq 0$,
$\bar{\eta}(0) = u_{ik}$ and rotate it clockwise, sweeping a full circle about $P_i$. As
$\bar{\eta}(\theta)$ goes from $\eta(0)$ to $\eta(2\pi)$ it will encounter
all the possible $u_{ij}$'s and these encounters define a sequence of angles
$\alpha_1,\alpha_2,\ldots,\alpha_r$ that add to $2\pi = \alpha_1 + \ldots +
\alpha_r$ ($\alpha_k =$ angle from k-th to (k+1)-th encounter with a $u_{ij}$,
$\alpha_r =$ angle from last encounter to first one again, see Figure
\ref{fig:field_of_view}).
If none of the angles $\{\alpha_1,\ldots,\alpha_r\}$ is bigger than $\pi$, set
$u_i^{+} = u_i^{-} = 0$. Otherwise define $u_i^{+} = u_{i(m)}$ and $u_i^{-} =
u_{i(n)}$ the unit vectors encountered when entering and exiting the angle
$\alpha_b > \pi$ bigger than $\pi$.

\begin{figure}[H]
  \centering
  \includegraphics[width=1.0\textwidth,natwidth=837,natheight=566]{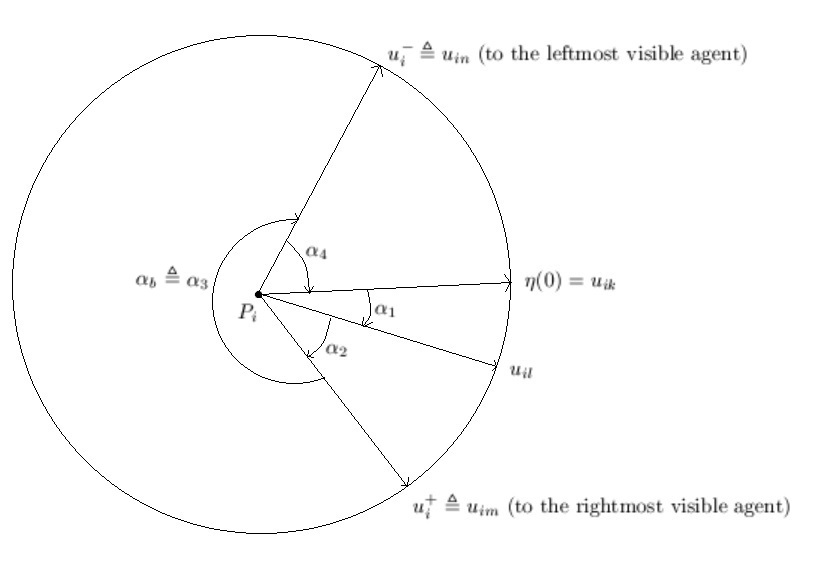}
  \caption{Leftmost and rightmost visible agents of agent located at $P_i$.}
  \label{fig:field_of_view}
\end{figure}

One might call $u_i^{-}$ the pointer to the ``leftmost visible agent'' from
$P_i$ and $u_i^{+}$ the pointer to the ``rightmost visible agent'' among the
neighbors of $P_i$. If $P_i$ has nonzero right and leftmost visible agents it
means that all its visible neighbors belong to a disk sector defined by an angle less than
$\pi$, and $P_i$ will be movable. Otherwise we call him ``surrounded'' by
neighbors and, in this case, it will stay in place while it remains
surrounded (see Appendix 2 for an alternative way of
defining the leftmost and rightmost agents).\\
The dynamics of the multi-agent system will be defined as follows. \\

\begin{equation} \label{eq:dynamics}
\frac{dP_i}{dt} = v_0 (u_i^{+} + u_i^{-})  \mbox{  for  }  i=1,\ldots,N
\end{equation}

Note that the speed of each agent is in the span of $[0,2v_0]$.\\
With this we have defined a local, distributed, reactive law of motion based on
the information gathered by each agent. Notice that the agents do not
communicate directly, are all identical, and have limited sensing capabilities,
yet we shall show that, under the defined reactive law of motion, in response to what
they can ``see'' (which is the bearings to their neighboring agents), the agents
will all come together while decreasing the distance between all pairs of visible
agents.\\
Assume that we are given an initial configuration of N agents placed in the
plane in such a way that their visibility graph is connected. This just means
that there is a path (or a chain) of mutually visible neighbors from each
agent to any other agent.\\
Our first result is that while agents move according to the above described
rule of motion, the visibility graph will only be supplemented with new edges
and old ``visibility connections'' will never be lost.\\

\subsection{Connectivity is never lost}

We shall show that\\

\begin{proposition} \label{prop:connectivity}
A multi agent systems under the dynamics 
\begin{displaymath}
\{\dot{P_i} = v_0 (u_i^{+} + u_i^{-})\}_{i=1,\ldots,N})
\end{displaymath}
ensures that pairs of neighboring agents at $t=0$ (i.e. agents at a distance
less than $V$) will remain neighbors forever.
\end{proposition}

\begin{proof}
To prove this result we shall consider the dynamics of distances between pairs
of agents. \\
We have that the distance $\Delta_{ij}$ between $P_i$ and $P_j$ is 
\begin{displaymath}
\Delta_{ij} = \Vert P_j - P_i \Vert = [(P_j-P_i)^T(P_j-P_i)]^{1/2}
\end{displaymath}
hence
\begin{displaymath}
\frac{d}{dt}\Delta_{ij}^{(t)} = \frac{1}{\Vert P_j - P_i \Vert}
(P_j-P_i)^T(\dot{P_j}-\dot{P_i})
\end{displaymath}
or
\begin{displaymath}
\begin{array}{ccc}
\frac{d}{dt}\Delta_{ij}^{(t)} & = & u_{ij}^T (\dot{P_j} - \dot{P_i}) \\
 & = & -u_{ij}^T \dot{P_i} + u_{ij}^T \dot{P_j} \\
 & = & -u_{ij}^T \dot{P_i} - u_{ji}^T \dot{P_j}
\end{array}
\end{displaymath}

But we know that the dynamics (\ref{eq:dynamics}) is

\begin{displaymath}
\begin{array}{ccc}
\dot{P_i} & = & v_0 (u_i^{+} + u_i^{-}) \\
\dot{P_j} & = & v_0 (u_j^{+} + u_j^{-})
\end{array}
\end{displaymath}

Therefore 

\begin{displaymath}
\frac{d}{dt}\Delta_{ij} = -v_0 u_{ij}^T (u_i^{+} + u_i^{-}) - v_0 u_{ji}^T
(u_j^{+} + u_j^{-})
\end{displaymath}

However for every agent $P_i$ we have either $u_i^{+} + u_j^{-}
\triangleq 0$ if agent is surrounded, or $u_i^{+} + u_i^{-}$ is in the direction
of the center of the disk sector in which all neighbors (including $P_j$)
reside (see Figure \ref{fig:scalar_product}).

\begin{figure}[H]
  \centering
  \includegraphics[width=0.8\textwidth,natwidth=427,natheight=319]{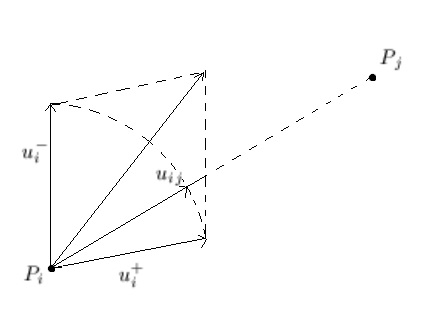}
  \caption{$u_{ij}^T(u_i^{+} + u_i^{-}) \geq 0$}
  \label{fig:scalar_product}
\end{figure}

Therefore the inner product $u_{ij}^T (u_i^{+} + u_i^{-}) = <u_{ij},(u_i^{+} +
u_i^{-})>$ will necessary be positive (see Appendix 3 for a formal proof),
hence

\begin{displaymath}
\frac{d}{dt}\Delta_{ij}^{(t)} = -( v_0 * positive + v_0 *positive ) \leq 0 
\end{displaymath}

This shows that distances between neighbors can only decrease (or remain the
same). Hence agents never lose neighbors under the assumed dynamics.

\end{proof}

\subsection{Finite-time gathering}

We have seen that the dynamics of the system (\ref{eq:dynamics}) ensures that
agents that are neighbors at $t=0$ will forever remain neighbors. We shall next
prove that, as time passes, agents acquire new neighbors and
in fact will all converge to a common point of encounter. We prove the following.\\

\begin{proposition} \label{prop:gathering}
A multi-agent system with dynamics given by
(\ref{eq:dynamics}) gathers all agents to a point in $\mathbb{R}^2$, in finite
time.
\end{proposition}

\begin{proof}
We shall rely on a Lyapunov function $L(P_1,\ldots,P_N)$, a positive
function defined on the geometry of agent constellations which becomes zero if
and only if all agents' locations are identical. We shall show that,
due to the dynamics of the system, the function $L(P_1,\ldots,P_N)$ decreases to
zero at a rate bounded away from zero, ensuring finite time convergence.\\
The function $L$ will be defined as the perimeter of the convex hull of all
agents' locations, $CH\{P_i(t)\}_{i=1,\ldots,N}$. Indeed, consider the set of
agents that are, at a given time $t$, the vertices of the convex hull of the set
$\{P_i(t)\}_{i=1,\ldots,N}$. Let us call these agents $\{\tilde{P}_k(t)\}$ for
$k=1,\ldots,K\leq N$. For every agent $\tilde{P}_k$ on the convex hull (i.e. for
every agent that is a corner of the convex polygon defining the convex hull), we
have that all other agents, are in a region (wedge) determined by the half lines
from $\tilde{P}_k$ in the directions $\tilde{P}_k\tilde{P}_{k-1}$ and
$\tilde{P}_k\tilde{P}_{k+1}$, a wedge with an opening angle say $\theta_k$ (see
Figure \ref{fig:polygon_angle}).

\begin{figure}[H]
  \centering
  \includegraphics[width=0.7\textwidth,natwidth=563,natheight=383]{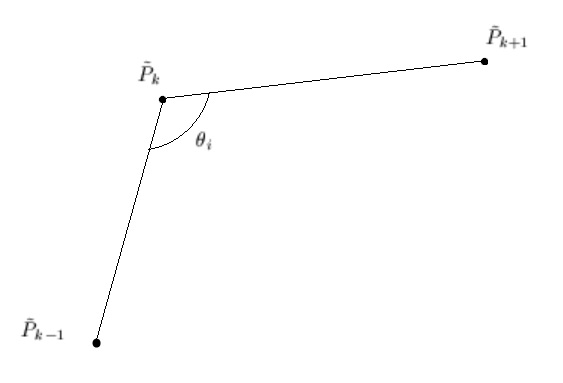}
  \caption{Angle at a vertex of the convex hull}
  \label{fig:polygon_angle}
\end{figure}

Since clearly $\theta_k \leq \pi$ for all $k$ we must have that agent
$\tilde{P}_k$ has all its visible neighbors in a wedge of his visibility disk
with an angle $\alpha_k \leq \theta_k \leq \pi$ hence his $u_k^{+}$ and
$u_{k}^{-}$ vectors will not be zero, causing the motion of $\tilde{P}_k$
towards the interior of the convex hull. This will ensure the shrinking of the
convex hull, while it exists, and the rate of this shrinking will be determined
by the evolution of the constellation of agents' locations. Let
us formally prove that indeed, the convex hull will shrink to a point in finite time. \\
Consider the perimeter $L(t)$ of $CH\{P_i(t)\}_{i=1,\ldots,N}$

\begin{displaymath}
L(t) = \sum_{k=1}^{K(t)} \Delta_{k,k+1} = \sum_{k=1}^{K(t)}
[(\tilde{P}_{k+1})(t) -
\tilde{P}_k(t))^T(\tilde{P}_{k+1}(t)-\tilde{P}_k(t))]^{1/2}
\end{displaymath}

where the indices are considered modulo $K(t)$.

We have, assuming that $K$ remains the same for a while,

\begin{displaymath}
\frac{d}{dt} L(t) = \sum_{k=1}^{K} \frac{d}{dt} \Delta_k = - \sum_{k=1}^K \left(v_0 \tilde{u}_{k,k+1}^T(u_k^{+} + u_k^{-}) +
 v_0\tilde{u}_{k,k+1}^T(u_{k+1}^{+}+u_{k+1}^{-})\right)
\end{displaymath}

but note that $\tilde{u}_{k,k+1}$ does not necessarily lie between
$u_k^{+}$ and $u_k^{-}$ anymore, since, in fact, $\tilde{P}_k$ and
$\tilde{P}_{k+1}$ might not even be neighbors. \\

Now let us consider $\frac{d}{dt}L(t)$ and rewrite it as follows

\begin{displaymath}
\frac{d}{dt}L(t) = - v_0 \sum_{k=1}^{K} \tilde{u}_{k,k+1}^T(u_k^{+}+u_k^{-})
- v_0 \sum_{k=1}^{K} \tilde{u}_{k+1,k}^T(u_{k+1}^{+}+u_{k+1}^{-})
\end{displaymath}

Rewriting the second term above, by moving the indices $k$ by -1 we get

\begin{displaymath}
\frac{d}{dt}L(t) = -v_0 \sum_{k=1}^{K} \tilde{u}_{k,k+1}^T(u_k^{+}+u_k^{-}) -
v_0 \sum_{k=1}^{K} \tilde{u}_{k,k-1}^T(u_{k}^{+}+u_{k}^{-})
\end{displaymath}

This yields 

\begin{displaymath}
\frac{d}{dt}L(t) = - v_0 \sum_{k=1}^{K} <u_k^{+}, \tilde{u}_{k,k+1} +
\tilde{u}_{k,k-1}> - v_0 \sum_{k=1}^{K} <u_k^{-}, \tilde{u}_{k,k+1} +
\tilde{u}_{k,k-1}>
\end{displaymath}

Note that we have here inner products between unit vectors, yielding the cosines
of the angles between them. Therefore, defining $\theta_k = $ the angle between
$\tilde{u}_{k,k-1}$ and $\tilde{u}_{k,k+1}$ (i.e. the interior angle of the
convex hull at the vertex k, see Figure \ref{fig:angles}), and the angles :

\begin{displaymath}
\begin{array}{c}
\alpha_k^{+} \triangleq \gamma(u_k^{+}, \tilde{u}_{k,k+1}) \\
\beta_{k}^{+} \triangleq \gamma(\tilde{u}_{k,k-1}, u_{k}^{+})\\
\alpha_k^{-} \triangleq \gamma(\tilde{u}_{k,k-1}, u_k^{-}) \\
\beta_k^{-} \triangleq \gamma(u_k^{-}, \tilde{u}_{k,k+1})
\end{array}
\end{displaymath}

\begin{figure}[H]
  \centering
  \includegraphics[width=1.0\textwidth,natwidth=770,natheight=408]{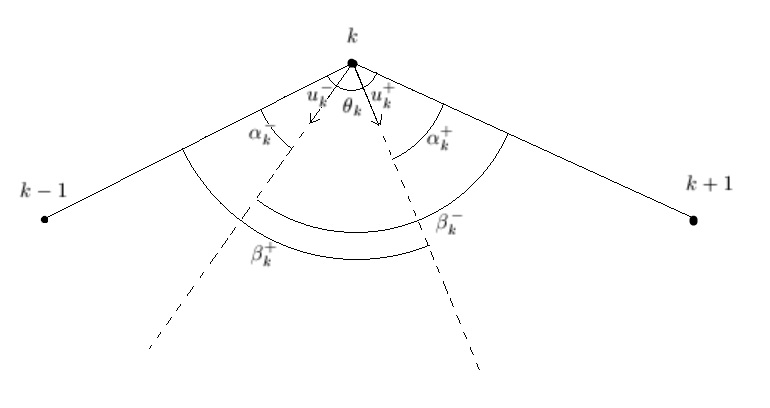}
  \caption{Angles at a vertex of the convex hull.}
  \label{fig:angles}
\end{figure}

we have $\alpha_k^{+} + \beta_k^{+} = \alpha_k^{-} + \beta_k^{-} =
\theta_k$ and all these angles are between $0$ and $\pi$.\\

Using these angles we can rewrite

\begin{displaymath}
\frac{d}{dt}L(t) = -\sum_{k=1}^{K}v_0(\cos\alpha_k^{+} + \cos\beta_k^{+}) -
\sum_{k=1}^{K} v_0 (\cos\alpha_k^{-} + \cos\beta_k^{-})
\end{displaymath}

Now, using the inequality (proved in Apenedix 1)

\begin{equation}\label{cosapluscosb}
\begin{array}{c}
\cos\alpha + \cos\beta \geq 1 + \cos(\alpha + \beta) \\
0 \leq \alpha, \beta, \alpha + \beta \leq \pi
\end{array}
\end{equation}

we obtain that 

\begin{equation}\label{eq:derivative_bound}
-\frac{d}{dt}L(t) \geq 2 v_0 \sum_{i=1}^{K} (1 + \cos\theta_i)
\end{equation}

For any convex polygon we have the following result (see the detailed proof
in Appendix 1) :

\begin{theorem}\label{polygon_theorem}
For any convex polygon with $K$ vertices and interior angles 
$\theta_1,\ldots,\theta_K$, with $(\theta_1+\ldots+\theta_K) = (K-2)\pi$ we have
that 
\begin{equation}\label{polygon_equation}
\sum_{k=1}^{K} \cos(\theta_i) \geq \left\{ \begin{array}{ccc}
				1 + (K-1)\cos\left(\frac{(K-2)\pi}{K-1}\right) & & 2 \leq K \leq 6 \\
				K\cos\left(\frac{(K-2)\pi}{K}\right) & & K \geq 7
				\end{array}
			\right.
\end{equation}
\end{theorem}
 
Therefore, we obtain from (\ref{eq:derivative_bound}) and
(\ref{polygon_equation}) that

\begin{equation} \label{bound_convergence}
-\frac{d}{dt}L(t) \geq \mu(K)
\end{equation}

where

\begin{displaymath}
\begin{array}{ccc}
\mu(K) & = & 2 v_0 \left(K + \left\{ \begin{array}{ccc} 1 +
(K-1)\cos\left(\frac{(K-2)\pi}{K-1}\right) & & 2 \leq K \leq 6 \\
				K\cos\left(\frac{(K-2)\pi}{K}\right) & & K \geq 7\end{array}
\right\} \right) \\
 & = & 2 v_0 K \left( 1 -
\max\left\{\cos\left(\frac{2\pi}{K}\right),\frac{K-1}{K}\cos\left(\frac{\pi}{K-1}\right)-
\frac{1}{K}\right\}\right)
\end{array}
\end{displaymath}

Note here that, since $(1 - max\{\ldots\}) > 0$ we have that the rate of
decrease in the perimeter of the configuration is srictly positive while the convex hull
of the agent location is not a single point.\\

The argument outlined so far assumed that the number of agents determining the
convex hull of their constellation is a constant $K$. Suppose however that in
the course of evolution some agents collide and/or some agents become
``exposed'' as vertices of the convex hull, and hence $K$ may jump to some
different integer value.
At a collision between two agents we assume that they merge and thereafter continue to
move as a single agent. Since irrespective of the value of $K$ the perimeter
decreases at a rate which is strictly positive and bounded away from zero we
have effectively proved that in finite time the perimeter of the convex hull
will necessarily reach 0. This concludes the proof of Theorem
\ref{prop:gathering}.
\\

\end{proof}

Figure \ref{fig:bound_plot} shows the bound as a function of $K$ assuming
$v_0 = 1$. Note that we always have $K \leq N$, and $\mu(K)$ is a decreasing
function of $K$, hence for any finite number of agents there will be a strictly
positive constant $\mu(N)$ so that

\begin{displaymath}
-\frac{d}{dt}L(t) \geq \mu(N)
\end{displaymath}

ensuring that after a finite time of $T_{ub}$ given by 

\begin{displaymath}
\mu(N) T_{ub} = L(0) \Rightarrow T_{ub} = \frac{L(0)}{\mu(N)}
\end{displaymath}

we shall have that $L(T_{ub}) = 0$. \\
Hence we have an upper bound on the time of convergence for any configuration of
$N$ agents given by $\frac{L(0)}{\mu(N)}$. \\

Note from (\ref{eq:derivative_bound}) that the rate of decrease does not
depend on the perimeter of the convex hull but only on the number of agents
forming it. This was an expected result, since the dynamics does not rely on
Euclidian distances.
This bound is decreasing with $K$, so that the
more agents form the convex hull, the smaller will be the rate of decrease. For
$K=2$ and $K=3$ the bound is $-8 v_0$, for $K=4$, it is $-7 v_0$, and then it keeps increasing for higher values of
$K$, slowly converging towards 0 form below. Note the change of curve between
$K=6$ and $K=7$, due to the ``interresting'' discontinuity in the geometric
result exhibited in equation (\ref{polygon_theorem}).

\begin{figure}[H]
  \centering
  \includegraphics[width=1.0\textwidth,natwidth=1050,natheight=329]{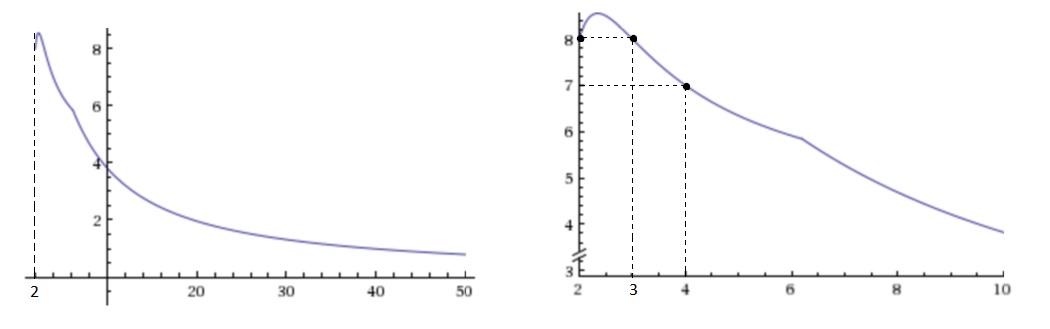}
  \caption{Graph of the bound $\mu(K)$ of (\ref{eq:derivative_bound}). The
  graph on the right is a zoom on small values of $K$.}
  \label{fig:bound_plot}
\end{figure}

The inequalities of (\ref{cosapluscosb}) and of
(\ref{polygon_equation}) become equalities for particular configurations of the
agents (for example a regular polygon in which each pair of adjacent neighbors
are visible to each other, if $K \geq 7$). In this case,
the bound in (\ref{eq:derivative_bound}) will yield the exact rate of
convergence of the convex-hull perimeter as long as $K$ remains the same.
\\

\section{Generalizations}

All the above analysis can be generalized for dynamics of the form

\begin{equation} \label{eq:dynamics2}
\frac{dP_i}{dt} = f(P^{(i)}) (u_i^{+} + u_i^{-})  \mbox{  for  }  i=1,\ldots,N
\end{equation}

$f(P^{(i)}) \geq 0$ is some positive function of the
configuration of the neighbours seen by agent $i$. This generalization also
guarantees that the rule of motion is locally defined and reactive, and
defined in the same way for all agents.\\
The dynamics (\ref{eq:dynamics}) correspond to a particular
case of (\ref{eq:dynamics2}), with $f(P^{(i)}) = v_0 = constant$ for all agents.\\

It is easy to slightly change the proofs above in order to show that Theorem
\ref{prop:connectivity} (ensuring that connectivity is not lost) is still valid
as long as $f(P^{(i)}) \geq 0$ for all $i$, and that Theorem \ref{prop:gathering}
(ensuring finite time gathering) is also valid as long as $f(P^{(i)}) \geq \epsilon
> 0$ for all $i$.\\

Note that in the work of Gordon et al \cite{Gordon1}, a constant speed for the
agents was considered, and this corresponds to setting $f(P^{(i)}) =
\frac{1}{|| u_i^{+} + u_i^{-}||}$ for a mobile agent $i$, rather than $v_0$.
Given that in this case $f(P^{(i)}) \geq \frac{1}{2}$, the conditions for
Theorem \ref{prop:connectivity} and Theorem \ref{prop:gathering} are verified, and
hence the dynamics with constant speed also ensures convergence to a single
point without pairs of initially visible agents losing connectivity.  We therefore also have a
proof for the convergence of the algorithm that was proposed in the
above-mentioned paper.

However, it was pointed out in the above-mentioned paper that in the model with
constant speed agents of one has to deal with quite unpleasant chattering effects and ``Zeno''-ness in order
to effectively modulate the speed of motion of some of the agents. In fact, in
certain configurations, an agent may ``oscillate'' infinitely often between a position in which it is ``surrounded'' to a position in which it is
not. This implies alternating its speed infinitely often between zero and a
constant value. In contrast, in our model, the speed is defined to vary
smoothly in the range $[0 , 2v_0]$. Therefore our model presents
some clear advantages and natural modulation for the speed of the
agents.

\section{Simulations}

Let us consider some simulations of the multi-agent dynamics discussed in this
paper. We will start with a randomly generated swarm and then we shall have
a look at some interesting particular configurations.

\subsection{Some random generated swarm of 15 agents}

First, we randomly generate a swarm of 15 agents, with a connected visibility
graph as initial configuration. We set $v_0$ to 1 and visibility to 200.
The configuration of the swarm at different times during the evolution is given
in Figure \ref{fig:random15_timeline}.

\begin{figure}[H]
  \centering
  \includegraphics[width=1\textwidth,natwidth=1306,natheight=1260]{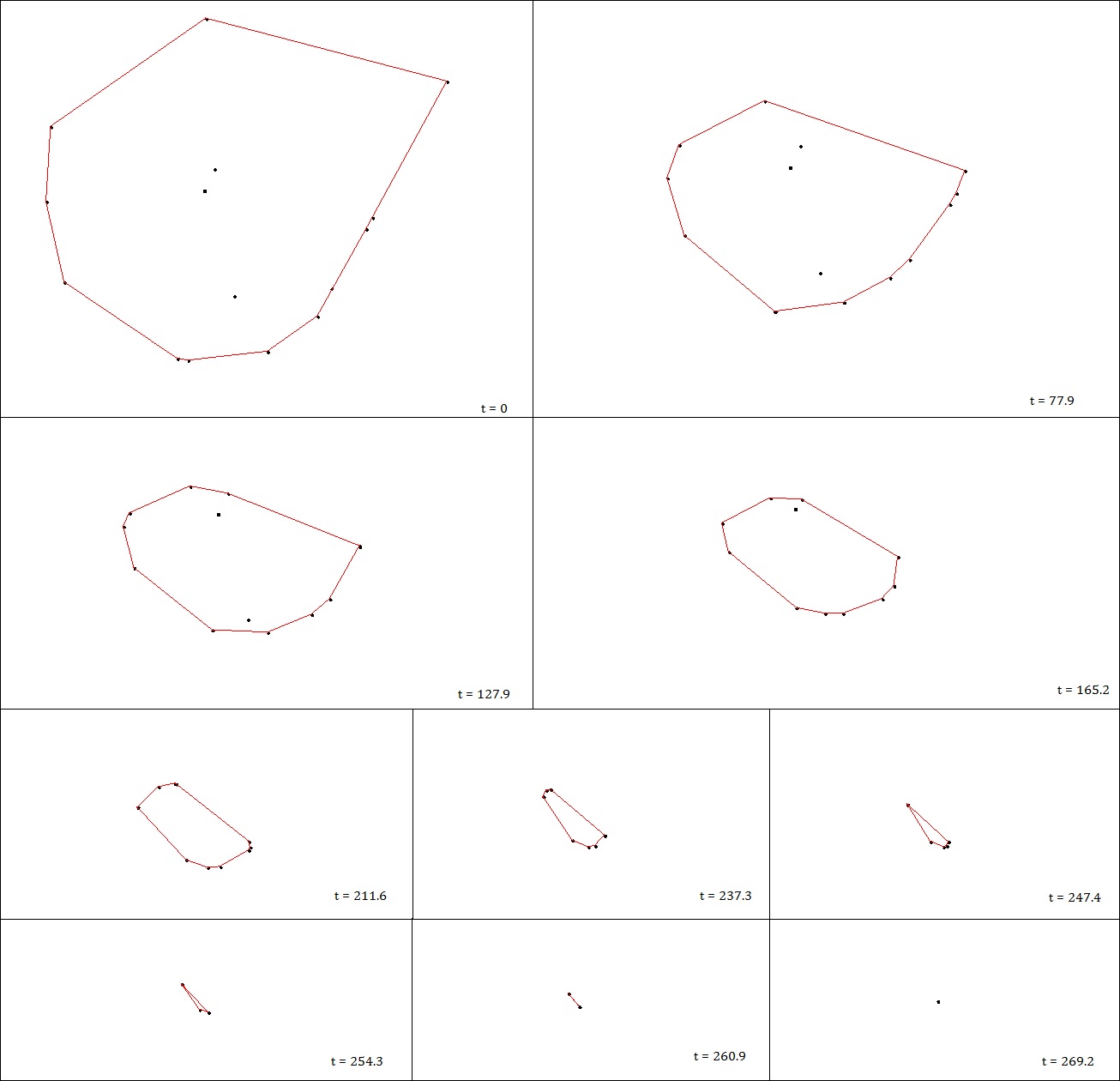}
  \caption{Configuration of the swarm at different times. The convex hull of
  the set of agents is also shown.}
  \label{fig:random15_timeline}
\end{figure}

We also plot some properties of the swarm during the time of its evolution.
Figure \ref{fig:random15_hull_perimeter} represents the perimeter of the convex hull of
the set of agents. Figure \ref{fig:random15_hull_n} plots the count of
indistinguishable agents (i.e. collided agents count as one) in the convex hull.
Figure \ref{fig:random15_hull_der} represents the time derivative of the convex
hull perimeter and the theoretical bound given by equation
(\ref{bound_convergence}), as function of the number of indistinguishable agents
in the convex hull.

\begin{figure}[H]
  \centering
  \includegraphics[width=0.8\textwidth,natwidth=724,natheight=463]{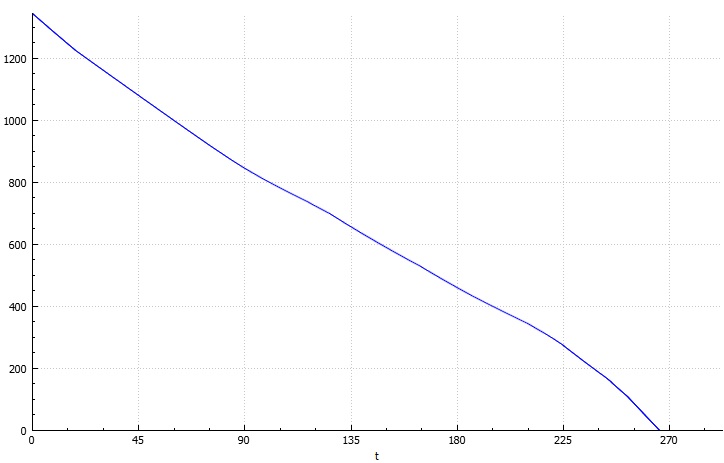}
  \caption{Perimeter of the convex hull, decreasing until it reaches zero.}
  \label{fig:random15_hull_perimeter}
\end{figure}

\begin{figure}[H]
  \centering
  \includegraphics[width=0.8\textwidth,natwidth=733,natheight=432]{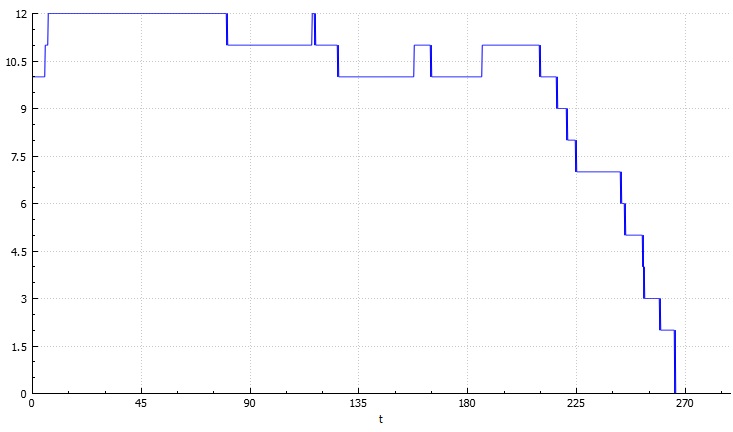}
  \caption{Number of agents forming the convex hull of the set. Decreases when
  two or more agents of the hull ``merge'' and increases when one or more agents
  are ``collected" by the hull and added to it.}
  \label{fig:random15_hull_n}
\end{figure}

\begin{figure}[H]
  \centering
  \includegraphics[width=0.8\textwidth,natwidth=778,natheight=326]{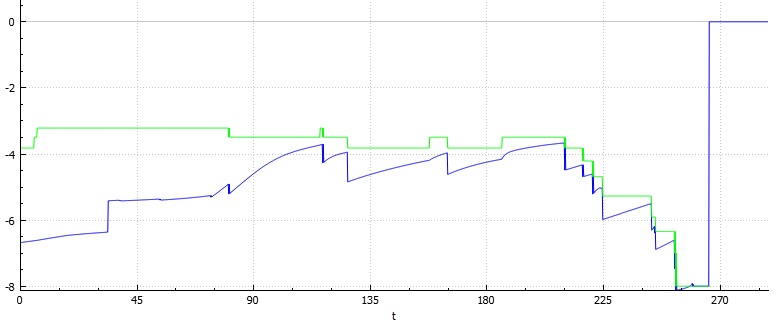}
  \caption{Derivative of the convex hull perimeter (blue) and theoretical bound
  (green), function of number of agents forming the convex hull, given by
  (\ref{bound_convergence}).}
  \label{fig:random15_hull_der}
\end{figure}

For a fixed number of agents forming the convex hull, one can see in Figure
\ref{fig:random15_hull_der} that the derivative of the perimeter goes towards
the theoretical bound. This can be intuitively explained by the facts that far
away agents evolve towards the inside more rapidly, making the convex hull
shape more regular, approaching to the
shape that yields the theoretical bound. \\

The discontinuity of the derivative of the perimeter of the convex hull
that occurs when there is no change in the number of agents forming the
hull, for example around $t=35$ in Figure \ref{fig:random15_hull_der}, is due to
change in the connectivity graph (which has not been printed for clarity). In
this particular case, two agents on the top left became visible to each other at
this time and their directions and speed changed, slowing down slightly the
rate of perimeter decrease.

\subsection{Regular polygon of 10 agents}

The initial configuration is a regular polygon with 10 agents. Again, $v_0$ is
set to 1 and visibility is 200. As expected, the decreasing rate of the perimeter of
the convex hull is constant, and all the agents contribute to the convex hull
all along, until the very end where they collide and merge. Simulations and
the converging parameters' plots are seen in Figures
\ref{fig:regular10_timeline} and \ref{fig:regular10_hull}.

\begin{figure}[H]
  \centering
  \includegraphics[width=0.9\textwidth,natwidth=874,natheight=329]{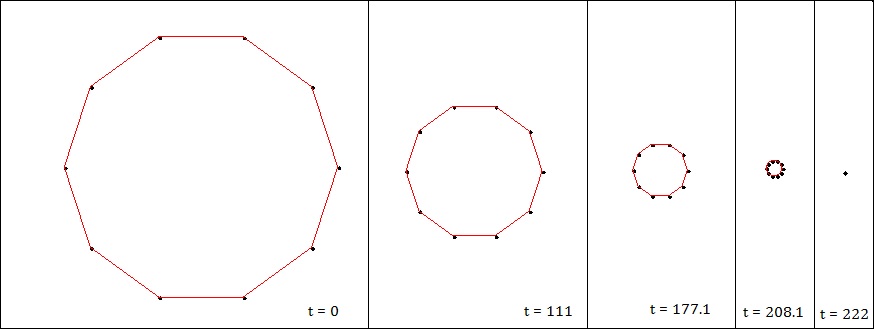}
  \caption{Starting as regular polygon, the swarm keeps its regular shape until
  the final collision. The leftmost and rightmost agents of each agent are its
  two neighbors in the polygon.}
  \label{fig:regular10_timeline}
\end{figure}

\begin{figure}[H]
  \centering
  \includegraphics[width=0.9\textwidth,natwidth=799,natheight=194]{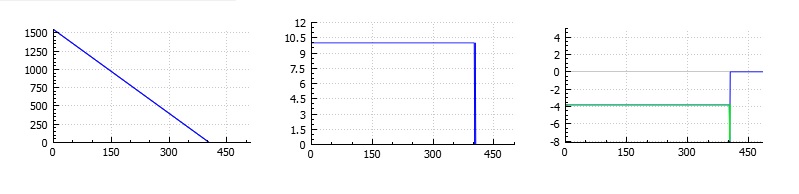}
  \caption{From left to right : the perimeter of the convex hull decreasing at
  a constant rate, the number of agents forming the convex hull being constant,
  and the derivative of the perimeter of the convex hull being constant
  and equaling its theoretical bound all along the evolution.}
  \label{fig:regular10_hull}
\end{figure}

\subsection{Close-to-minimum-configuration with $n=4$}

We start at a configuration close to the configuration that reaches the minimum
of the sum of the cosinuses of the interior angles of the polygon. Results are
represented in Figures \ref{fig:min4_timeline}, \ref{fig:min4_hull}, \ref{fig:min4_hull_der_zoom}. One
can see that this configuration provides a smaller decreasing rate than the
configuration of a regular polygon, in conformity with our analysis, where the
bound is not realized for a regular polygon configuration for small values of
$n$.

\begin{figure}[H]
  \centering
  \includegraphics[width=0.9\textwidth,natwidth=438,natheight=239]{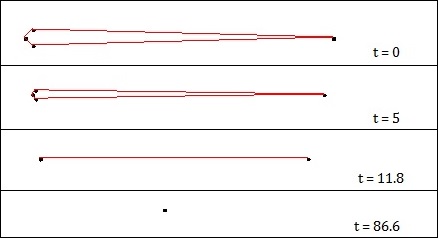}
  \caption{The cone configuration approaches the theoretical bound
  configuration. In this case, the gathering can be divided in two stages.
  First, the 3 agents of one side merge, then the resulting composite agent
  merges with the single agent of the other side.}
  \label{fig:min4_timeline}
\end{figure}

\begin{figure}[H]
  \centering
  \includegraphics[width=0.9\textwidth,natwidth=766,natheight=403]{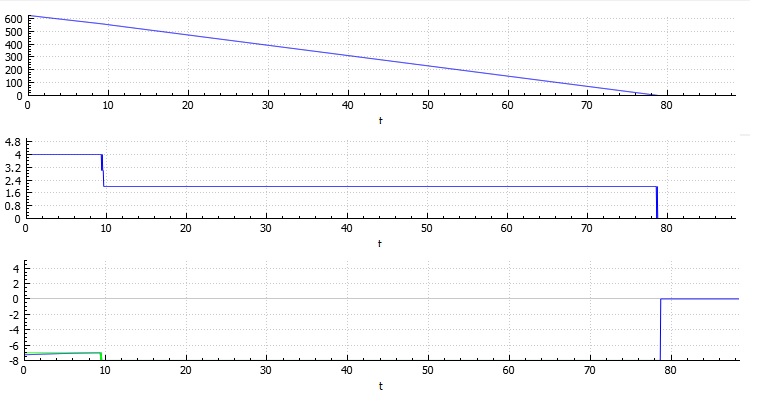}
  \caption{From top to bottom : The perimeter of the convex hull, the number of
  agents forming it, and the derivative of it.}
  \label{fig:min4_hull}
\end{figure}

\begin{figure}[H]
  \centering
  \includegraphics[width=0.9\textwidth,natwidth=674,natheight=269]{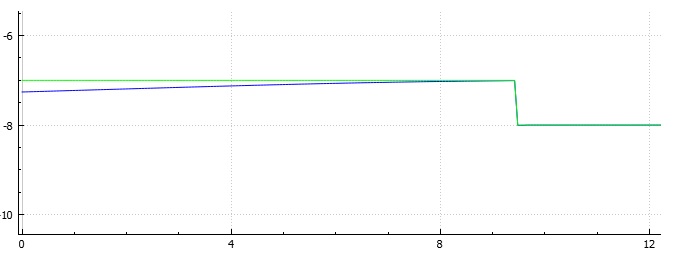}
  \caption{A zoom on the the derivative of the perimeter of the convex hull
  (in blue) at the beginning of the dynamics. A regular polygon (a square in
  this case) would give a constant rate of -8, lower than the current rate
  approaching the theoretical bound of -7 (in green).}
  \label{fig:min4_hull_der_zoom}
\end{figure}

\section{Concluding remarks}

We have shown that a very simple local control on the velocity of agents in the
plane, based on limited visibility and bearing only sensing of neighbors
ensures their finite time gathering. The motion rule is simply to set the
agents' velocity to equal the vector sum of unit vector pointers to two external
neighbors if all visible neighbors reside inside a halfplane (a half-disk) about
the agent, otherwise set the velocity to zero. This very simple rule of behavior
is slightly different from the one assumed by Gordon, Wagner, and Bruckstein in
\cite{Gordon1}, where the motion was set to have a constant velocity in the
direction bisecting the disk sector where visible neighbors reside, or zero if the agent was
``surrounded''. As we showed in this paper, that
model also ensures gathering. However, it was pointed out there that their
proposed model had to deal with some quite unpleasant chattering or
zeno-ness effects in order to effectively modulate the speed of
motion of agents. \\
In this paper, in conjunction with our model, and some generalizations too,
including the model of \cite{Gordon1}, we provided a very simple geometric proof
that finite time gathering is achieved, and provided precise bounds on the
rate of decrease of the perimeter of the agent configuration's convex hull.
These bounds are based on a geometric lower bound on the sum of cosines of the
interior angles of an arbitrary convex planar polygon, that is interesting on
its own right (a curious breakpoint occurring in the bound at 7 vertices). Our
result may be regarded as a convergence proof for a highly nonlinear autonomous
dynamic system, naturally handling dynamic changes in its dimension (the events
when two agents meet and merge).

\newpage

\bibliographystyle{unsrt}
\bibliography{references}

\newpage

\begin{center}
{\bf Appendix 1 : Proof of Lemma 1} \label{AppendixProof}
\end{center}

We will first prove the following facts :

\begin{fact}
\label{prop:cosaplusb}
Let $0 \leq a \leq b \leq \pi$ and $0 \leq a + b \leq \pi$. Then we have 

\begin{displaymath}
\sqrt{2(1+\cos(a+b))} = 2\cos\left(\frac{a+b}{2}\right) \geq \cos(a) + \cos(b)
\geq 2 \cos^2\left(\frac{a+b}{2}\right) = 1 + \cos(a+b)
\end{displaymath}
\end{fact}

\begin{proof}
The function cosine is decreasing in $[0,\pi]$, and given that $\frac{a+b}{2}
\geq \frac{b-a}{2}$ :

\begin{displaymath}
1 \geq \cos\left(\frac{b-a}{2}\right) \geq \cos\left(\frac{a+b}{2}\right)
\end{displaymath}

Multiplying by $2\cos\left(\frac{a+b}{2}\right) \geq 0$ :

\begin{displaymath}
\begin{array}{lllll}
2\cos\left(\frac{a+b}{2}\right) & \geq &
2\cos\left(\frac{a+b}{2}\right)\cos\left(\frac{b-a}{2}\right) & \geq & 2 \cos^2\left(\frac{a+b}{2}\right) \\
2\cos\left(\frac{a+b}{2}\right) & \geq & \cos(a) + \cos(b) & \geq & 1 +
\cos(a+b)
\end{array}
\end{displaymath}

\end{proof}

A direct consequence is the following lemma.

\begin{fact}
\label{prop:cosaplusb2}
Let $0 \leq a,b \leq \pi$. Then 

\begin{displaymath}
\cos(a) + \cos(b) \geq \left\{ 
	\begin{array}{lr} 
		1 + \cos(a+b) & : a+b \leq \pi \\
		2\cos\left(\frac{a+b}{2}\right) & : a+b \geq \pi 	
	\end{array}
\right.
\end{displaymath}

\end{fact}

\begin{proof}
The first line is already part of Lemma \ref{prop:cosaplusb}. The second
line can be proven by using the left inequality of Lemma
\ref{prop:cosaplusb} with $\pi - a$ and $\pi -b$, noticing that $0 \leq \pi - a
\leq \pi$, $0 \leq \pi - b \leq \pi$, and $\pi - a + \pi - b \leq \pi$ for $a+b
\geq \pi$.
\end{proof}

Now we can prove Lemma \ref{polygon_theorem}. Suppose any given
initial configuration of the polygon with interior angles $0
\leq x_1,\ldots,x_n \leq \pi$.
We then have $x_1 + \ldots + x_n = (n-2)\pi$. \\

Now repeat the following step : Go through all the pairs of non-zero values
($x_i,x_j$). As long as there is still a pair verifying $x_i + x_j
\leq \pi$, transform it from $(x_i,x_j)$ to $(0,x_i+x_j)$. When there are no
such pairs, then among all the non-zero values, take the the minimum
value and the maximum value, say $x_i$ and $x_j$ (they must verify $x_i + x_j
\geq \pi$ due to the previously applied process) , and transform the pair from
$(x_i,x_j)$ to $\left(\frac{x_i+x_j}{2},\frac{x_i+x_j}{2}\right)$. \\

Repeat the above process until convergence. We prove that the process
converges and that we can get as close as wanted to a configuration where all
non-zero values are equal. Note that after each step, the sum of the values is
unchanged, $(n-2)\pi$, and that the values of all $x_i$'s remain between $0$ and
$\pi$.\\

The number of values that the above process set to zero must be less or
equal to 2 in order to have the sum of the $n$ positive values equal to
$(n-2)\pi$. Therefore we can be sure that after a finite number of iterations,
there will be no pairs of nonzero values whose sum will be less than $\pi$
(otherwise this would allow us to add a zero value without changing the sum).\\

Once in this situation, all we do is replacing pairs of ``farthest" non-zero
values $(x_i,x_j)$ with the pair $\left(\frac{x_i+x_j}{2},\frac{x_i+x_j}{2}\right)$.
Let us show that all the nonzero values will converge to the same value,
specifically to their mean.\\

Let $k$ be the number of remaining non-zero values after the iteration $t_0$
which sets the ``last value'' to zero. Denote these values at the i-th
iteration by $({x_1^{(i)}, \ldots, x_k^{(i)}})$.
Define :

\begin{displaymath}
m = \frac{x_1^{(i)} + \ldots + x_k^{(i)}}{k} = \frac{(n-2)\pi}{k}
\end{displaymath}

\begin{displaymath}
E_i = (x_1^{(i)} -  m)^2 + \ldots + (x_k^{(i)} -  m)^2
\end{displaymath}

Suppose, without loss of generality, that at the $i$-th iteration the extreme
values were $x_1$ and $x_2$ and so we transformed $(x_1^{(i)},x_2^{(i)})$ into
$\left(x_1^{(i+1)} = \frac{x_1^{(i)}+x_2^{(i)}}{2}, x_2^{(i+1)} =
\frac{x_1^{(i)}+x_2^{(i)}}{2}\right)$. So we have :

\begin{displaymath}
\begin{array}{ccc}
E_{i+1} - E_i & = & 2(\frac{x_1^{(i)} + x_2^{(i)}}{2} - m)^2 - (x_1^{(i)} - m)^2
- (x_2^{(i)} - m)^2) \\
 & = & -\frac{1}{2}(x_1^{(i)} - x_2^{(i)})^2
\end{array}
\end{displaymath}

But $x_1^{(i)}$ and $x_2^{(i)}$ being the extreme values, we have for any $1
\leq l \leq k$ :

\begin{displaymath}
(x_1^{(i)} - x_2^{(i)})^2 \geq (x_l^{(i)} - m)^2
\end{displaymath}

and by summing over $l$ we get that : 

\begin{displaymath}
k (x_1^{(i)} - x_2^{(i)})^2 \geq E_i
\end{displaymath}

Hence

\begin{displaymath}
\begin{array}{ccccc}
 & & E_{i+1} - E_i & = & -\frac{1}{2}(x_1^{(i)} - x_2^{(i)})^2 \leq
 -\frac{E_i}{2k}
\\
 & & E_{i+1} & \leq & \left(1 - \frac{1}{2k}\right) E_i \\
0 & \leq & E_i & \leq & \left(1 - \frac{1}{2k}\right)^{i-t_0} E_{t_0}
\end{array}
\end{displaymath}

proving that $E_i$ converges to zero, i.e. all the non-zero values converge to
$m$.\\

At each step of the above described process, according to fact
\ref{prop:cosaplusb2}, the sum of cosines can only decrease. Therefore from
any given configuration we can get as close as possible to a configuration in
which all non-zero values are equal, without increasing the sum of the
cosines. Hence, the minimum value must be reached in a configuration in which
all non-zero values are equal. \\

Remebering that there can be at most only two zero values, the minimum
value of the sum of the cosines is the minimum of the following :

\begin{itemize}
  \item $2 + (n-2)\cos\left(\frac{(n-2)\pi}{n-2}\right) = -(n-4)$ (case with 2
  zeros)
  \item $1 + (n-1)\cos\left(\frac{(n-2)\pi}{n-1}\right)$ (case with 1 zero)
  \item $n\cos\left(\frac{(n-2)\pi}{n}\right)$ (case with no zero)
\end{itemize}

Now let us compare these values. In order to do so define the function
$e(x) = \cos(x) + x\sin(x)$ for $0 \leq x \leq \pi$.
Basic calculations give us $e^{'}(x) = x\cos(x)$ and therefore $e$ is increasing in
$\left[0,\frac{\pi}{2}\right]$ and decreasing in $\left[\frac{\pi}{2},
\pi\right]$. Therefore for $ 0 \leq x \leq \frac{\pi}{2}$, $e(x) \geq e(0) = 1$

In order to compare the case with 2 zeros to the case with 1 zero, define for
$n \geq 2$

\begin{displaymath}
\begin{array}{lll}
f(n) & = & 1 + (n-1)\cos\left(\frac{(n-2)\pi}{n-1}\right) + (n-4) \\
 & = & n - 3 - (n-1)\cos\left(\frac{\pi}{n-1}\right) \\
f^{'}(n) & = & 1 - \cos\left(\frac{\pi}{n-1}\right) -
\frac{\pi}{n-1}\sin\left(\frac{\pi}{n-1}\right) \\
 & = & 1 - e\left(\frac{\pi}{n-1}\right)
\end{array}
\end{displaymath}

Therefore $f^{'}(n) \leq 0$ for $n \geq 4$ because $\frac{\pi}{n-1} \leq
\frac{\pi}{2}$\\

$f(2) =  0 \leq 0$\\
$f(3) = 0 \leq 0$\\
$f(4) = -\frac{1}{2} \leq 0 $ \\
$f(n) \leq f(4) \leq 0$ for $n \geq 4$\\

Therefore $f(n) \leq 0$ and the case with 2 zeros is never exclusively
the optimal solution (since the case with 1 zero always has a smaller
or an equal value).

In order to compare the two remaining values, define for $n \geq 2$

\begin{displaymath}
h(n) = n\cos\left(\frac{2\pi}{n}\right) - (n-1)\cos\left(\frac{\pi}{n-1}\right)
+ 1
\end{displaymath}

The derivative according to $n$ is

\begin{displaymath}
\begin{array}{lll}
h^{'}(n) & = & \cos\left(\frac{2\pi}{n}\right) +
\frac{2\pi}{n}\sin\left(\frac{2\pi}{n}\right) - \cos\left(\frac{\pi}{n-1}\right)
- \frac{\pi}{n-1}\sin\left(\frac{\pi}{n-1}\right) \\
 & = & e\left(\frac{2\pi}{n}\right) - e\left(\frac{\pi}{n-1}\right)
\end{array}
\end{displaymath}

For $n \geq 4$, we have $\frac{\pi}{n-1} \leq \frac{2\pi}{n} \leq \frac{\pi}{2}$
and therefore $h^{'}(n) \geq 0$. One can check that $h(n) \leq 0$ for $2 \leq n
\leq 6$, and $h(7) > 0$, therefore $h(n) > 0$ for $n \geq 7$. This allows us to
conclude that for $2 \leq n \leq 6$ the minimal configuration is the one
corresponding to 1 zero, whereas for $n \geq 7$, it is the one with no zeros.

\newpage

\begin{center}
{\bf Appendix 2}
\end{center}

Let us understand the geometric Lemma \ref{polygon_theorem} by some
illustrations.
Consider a convex polygon of $n$ vertices. The sum of its interior
angles is equal to $(n-2)\pi$ and each angle is between $0$ and $\pi$. Denote by
$C_n$ the bound given by Lemma \ref{polygon_theorem}.

First let us notice that sometimes the minimum value of $C_n$ corresponds to a
set of interior angles that cannot be realized by a polygon in the plane (for example we can show that in a convex
polygon, if one of the interior angles is $0$, then there are exactly two
interior angles of $0$ and $(n-2)$ interior angles of $\pi$, but for example
the configuration that realizes the minimum of $C_4$ with $n=4$ does not
correspond to such a realizable configuration, see figure
\ref{fig:polygon_bound_n4} for an illustration).
However, we can get arbitrarily close to this value of $C_n$ with changing the angles with value of
$0$ by some $\epsilon$ and substracting in other angles (that are not $0$) the
added values (see figures \ref{fig:polygon_bound_approximated_n3} and
\ref{fig:polygon_bound_approximated_n4} for examples).
\\

Let us explain the intuition behind Lemma \ref{polygon_theorem}.
If ones tries to minimize the cosinus of an angle, he would
open it at maximum. 
But the constraint of forming a convex polygon forces one to close the loop
of the polygon. This is what limits how much one can open the angles. It
is intuitive that the more agents we have, the more freedom we have to open the
angles and use the numerous agents at our disposal to close the loop. \\

The abstract case when there is an infinity of agents corresponds to a circle
where all angles can be open at maximum, i.e. with an angle of $\pi$.\\

What this Lemma also infers is that the configuration of the polygon that
reaches the minimum value is a regular polygon for $n \geq 7$, and that for $n
\leq 6$, the minimum value of the sum of the cosinuses of the interior angles
can be arbitrarily closely approached by a polygon in a shape of a cone. Figures
\ref{fig:polygon_bound_n3}, \ref{fig:polygon_bound_approximated_n3}, and \ref{fig:polygon_regular_n3} illustrate this for the case $n=3$, while figures
\ref{fig:polygon_bound_n4}, \ref{fig:polygon_bound_approximated_n4}, and
\ref{fig:polygon_regular_n4} illustrate it for the case $n=4$. The cases $n=5$
and $n=6$ are similar, and for $n=7$ and above, the configuration of the minimum
is a regular polygon leading the value of $C_n =
n\cos\left(\frac{(n-2)\pi}{n}\right)$ (see examples in figures
\ref{fig:polygon_bound_n7},\ref{fig:polygon_bound_n10},\ref{fig:polygon_bound_n30}
).

\begin{figure}[H]
  \centering
  \includegraphics[width=0.8\textwidth,natwidth=443,natheight=218]{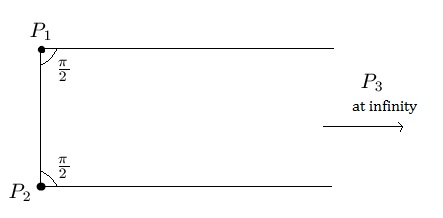}
  \caption{Theoretical configuration corresponding the the minimum value of
  $C_3 = 1$}
  \label{fig:polygon_bound_n3}
\end{figure}

\begin{figure}[H]
  \centering
  \includegraphics[width=0.8\textwidth,natwidth=622,natheight=210]{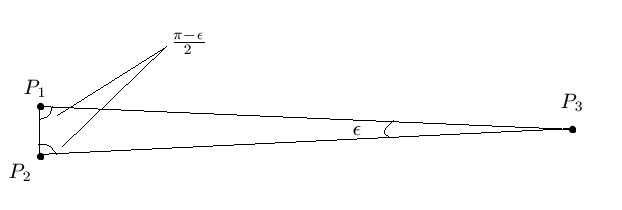}
  \caption{Practical configuration for $n=3$ provinding a value that can get
  potentially arbitrarily close to the theoretical minimum, with the help of $\epsilon$.}
  \label{fig:polygon_bound_approximated_n3}
\end{figure}

\begin{figure}[H]
  \centering
  \includegraphics[width=0.3\textwidth,natwidth=213,natheight=227]{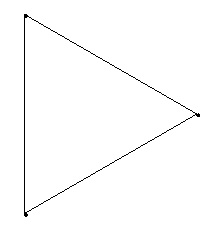}
  \caption{Regular polygon with $n=3$, the sum of cosinuses of the interior
  angles is $3\cos\left(\frac{\pi}{3}\right) = \frac{3}{2}$. This value is
  greater than the minimum possible value. }
  \label{fig:polygon_regular_n3}
\end{figure}

\begin{figure}[H]
  \centering
  \includegraphics[width=0.8\textwidth,natwidth=685,natheight=327]{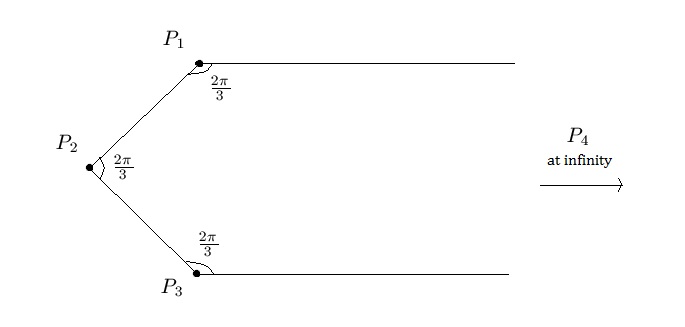}
  \caption{Theoretical configuration corresponding the the minimum value of
  $C_4 = -\frac{1}{2}$}
  \label{fig:polygon_bound_n4}
\end{figure}

\begin{figure}[H]
  \centering
  \includegraphics[width=0.8\textwidth,natwidth=733,natheight=197]{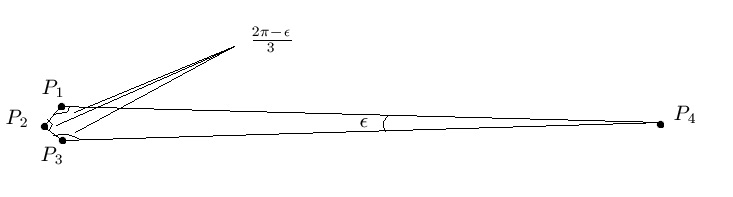}
  \caption{Practical configuration for $n=4$ provinding a value that can get
  potentially arbitrarily close to the theoretical minimum, with the help of $\epsilon$.}
  \label{fig:polygon_bound_approximated_n4}
\end{figure}

\begin{figure}[H]
  \centering
  \includegraphics[width=0.3\textwidth,natwidth=313,natheight=317]{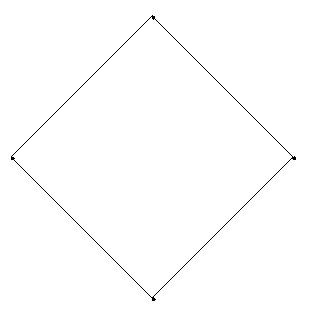}
  \caption{Regular polygon with $n=4$, the sum of cosinuses of the interior
  angles is $4\cos\left(\frac{2\pi}{4}\right) = 0$. This value is
  greater than the minimum possible value. }
  \label{fig:polygon_regular_n4}
\end{figure}

\begin{figure}[H]
  \centering
  \includegraphics[width=0.3\textwidth,natwidth=436,natheight=399]{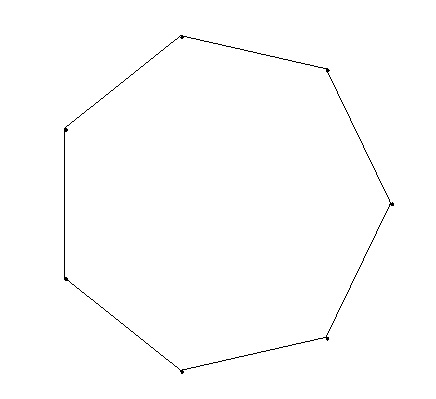}
  \caption{$n=7$ is the first value of $n$ for which the configuration of
  minimum value of the sum of the cosinuses of the interior angles
  corresponds to a regular polygon. In this case, this value is
  $7\cos\left(\frac{5\pi}{7}\right) \approx -4.36$. }
  \label{fig:polygon_bound_n7}
\end{figure}

\begin{figure}[H]
  \centering
  \includegraphics[width=0.3\textwidth,natwidth=585,natheight=519]{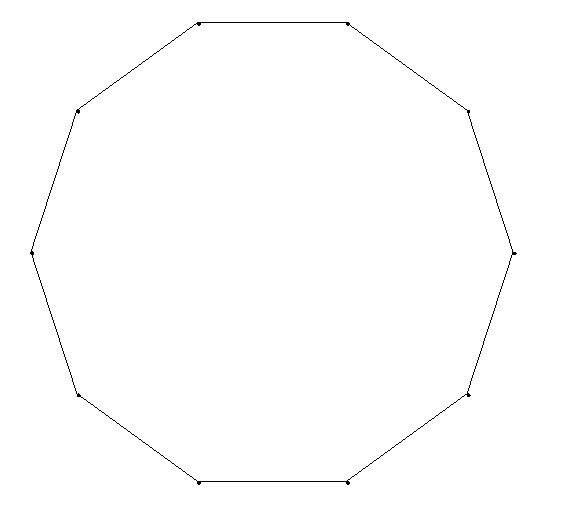}
  \caption{Regular polygon and configuration of minimum value for $n=10$, with
  value $10\cos\left(\frac{8\pi}{10}\right) \approx -8.09$, each
  angle contributing a value of $\cos\left(\frac{8\pi}{10}\right) \approx
  -0.81$.}
  \label{fig:polygon_bound_n10}
\end{figure}

\begin{figure}[H]
  \centering
  \includegraphics[width=0.3\textwidth,natwidth=585,natheight=519]{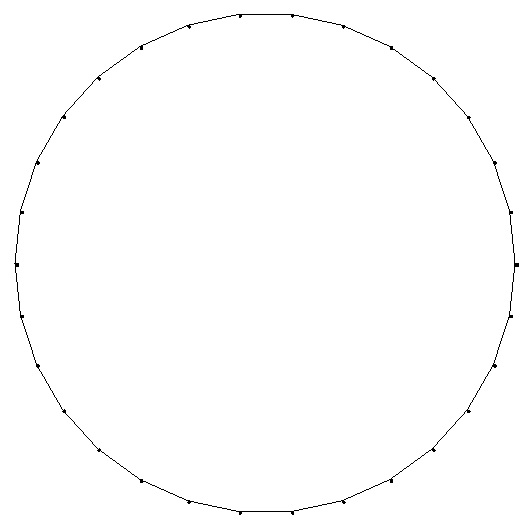}
  \caption{Regular polygon and configuration of minimum value for $n=30$, with
  value $30\cos\left(\frac{28\pi}{30}\right) \approx -29.34$, each
  angle contributing a value of $\cos\left(\frac{28\pi}{30}\right) \approx
  -0.98$. This last value gets closer to $-1$ with higher values of $n$ as
  explained.}
  \label{fig:polygon_bound_n30}
\end{figure}

\newpage
\begin{center}
{\bf Appendix 3}\label{AppendixAlternative}
\end{center}

The dynamics of the sytem described by (\ref{eq:dynamics}) can be
defined in the following alternative way.

We define the two following functions :

\begin{displaymath}
h^{+}(x) = \left\{ \begin{array}{lr} 
						1 & : x > 0 \\
						\frac{1}{2} & : x=0\\
						0 & : x < 0 
					\end{array}
			\right.					
\end{displaymath}

and $h^{-}(x) = h^{+}(-x)$.

Let $e_z$ be a unitary vector orthogonal to the plane (in any direction).
Then we define :

\begin{displaymath}
s_{ijk} = (u_{ij} \times u_{ik}) \cdot e_z 
\end{displaymath}

\begin{displaymath}
p_{ij}^{\pm} = \prod_{k} h^{\pm}(s_{ijk})
\end{displaymath}

\begin{displaymath}
w_{ij}^{\pm} = \frac{p_{ij}^{\pm}}{\sum_{j} p_{ij}^{\pm}}
\end{displaymath}

where $w_{ij}^{\pm} = 0$ if $\sum_{j} p_{ij}^{\pm} = 0$.

Finally, we define $w_{ij} = w_{ij}^{+} + w_{ij}^{-}$ and the equations of
movement are given by :

\begin{displaymath}
\dot{x}_{i} = v_0 \sum_{j} w_{ij} u_{ij}
\end{displaymath}

And the vectors $u_{i}^{+}$ and $u_{i}^{-}$ of system (\ref{eq:dynamics}) are
given by:

\begin{displaymath}
u_{i}^{\pm} = \sum_{j} w_{ij}^{\pm} u_{ij}
\end{displaymath}

\newpage

\begin{center}
{\bf Appendix 4}\label{AppendixScalarProduct}
\end{center}

Using the fact that $u_{i}^{-}$ and $u_{ij}$ are, by the way they are defined,
on the same half-plane of $u_{i}^{+}$

\begin{displaymath}
(u_{i}^{+} \times u_{ij}) \cdot (u_{i}^{+} \times u_{i}^{-}) \geq 0 
\end{displaymath}

where $\times$ is the cross product of vectors.
In the same way, $u_{ij}$ and $u_i^{+}$ are on the same half-plane of $u_i^{-}$
:

\begin{displaymath}
(u_{i}^{-} \times u_{ij}) \cdot (u_{i}^{-} \times u_{i}^{+}) \geq 0
\end{displaymath}

But, using the fact that these vectors are unit vectors,

\begin{displaymath}
\begin{array}{lll}
(u_{i}^{+} \times u_{ij}) \cdot (u_{i}^{+} \times u_{i}^{-}) & = & u_{ij}
\cdot u_{i}^{-} - (u_{i}^{+} \cdot u_{i}^{-}) (u_{ij} \cdot u_{i}^{+})\\
(u_{i}^{-} \times u_{ij}) \cdot (u_{i}^{-} \times u_{i}^{+}) & = & u_{ij}
\cdot u_{i}^{+} - (u_{i}^{-} \cdot u_{i}^{+}) (u_{ij} \cdot u_{i}^{-})
\end{array}
\end{displaymath}

Therefore

\begin{equation}\label{eq:between}
(u_{i}^{+} \times u_{ij}) \cdot (u_{i}^{+} \times u_{i}^{-}) + (u_{i}^{-} \times
u_{ij}) \cdot (u_{i}^{-} \times u_{i}^{+}) = (1 - u_{i}^{+} \cdot
u_{i}^{-}) u_{ij} \cdot (u_{i}^{+} + u_{i}^{-})
\end{equation}

Now $u_{i}^{+} \cdot u_{i}^{-} = 1$ implies that $u_{i}^{+} = u_{i}^{-} =
u_{ij}$, and in this case $(u_{i}^{+} + u_{i}^{-}) \cdot u_{ij} = 2 > 0$.

In any other case, $1 - u_{i}^{+} \cdot u_{i}^{-} > 0$, and given that the left
hand side of (\ref{eq:between}) is positive, we must have 

\begin{displaymath}
(u_{i}^{+} + u_{i}^{-}) \cdot u_{ij} \geq 0
\end{displaymath}

\newpage

\end{document}